\pdfoutput=1
\documentclass[11pt,a4paper]{article}

\usepackage[colorlinks=true,linkcolor=red,citecolor=blue]{hyperref}
\usepackage{fullpage,paralist}
\usepackage{theorem}
\usepackage{xspace}
\usepackage{verbatim}
\usepackage{amssymb}
\usepackage{amsmath}

\newtheorem{theorem}{Theorem}[section]

\newtheorem{corollary}{Corollary}[section]
\newskip\theorempreskipamount
\newskip\theorempostskipamount
\global
\global

\newenvironment{proof}{\noindent {\bf Proof.  }}{\hfill$\Box$}
\newenvironment{proofof}[1]{\smallskip
\noindent {\bf Proof of #1.  }}{\hfill$\Box$}

\newtheorem{lemma}{Lemma}[section]
\newtheorem{claim}{Claim}[section]

\newcommand{\setdiff}{\backslash}

\newcommand{\reals}{\mathbb{R}\xspace}
\newcommand{\Sn}{S_n\xspace}

\DeclareMathOperator*{\argmin}{arg\,min\xspace}
\DeclareMathOperator*{\argmax}{arg\,max\xspace}

\newcommand{\dotp}[2]{#1 \dotp #2\xspace}

\newcommand{\frob}[1]{\|#1\|_F\xspace}

\newcommand{\trace}[1]{\operatorname{Tr}(#1)}
\newcommand{\tracepd}[2]{\operatorname{Tr}(#1,#2)}
\newcommand{\innerpd}[3]{\langle #1,#3 \rangle_{#2}}

\newcommand{\specA}{A = U \Lambda U^T}
\newcommand{\specB}{B = V \Gamma V^T}
\newcommand{\specprA}{A = U' \Lambda {U'}^T}

\newcommand{\NP}{\textsf{NP}\xspace}

\newcommand{\qap}{\textsc{QAP}\xspace}

\newcommand{\bigO}{\mathcal{O}}
\newcommand{\bigOstar}{\bigO^*}
\newcommand{\rtime}{$\bigOstar({n^{\mathcal{O}(kp^2)}})$\xspace}
\newcommand{\qvp}{\textsc{QVP}\xspace}
\newcommand{\realmat}[1]{\reals^{#1 \times #1}}

\usepackage{framed}
\usepackage{float}

\newcommand{\RR}{{\mathbb R}}
\newcommand{\dist}{\operatorname{dist}}
\newcommand{\SIM}{\textsc{GSim}}
\newcommand{\WSIM}{\textsc{WSim}}
\newcommand{\MSIM}{\textsc{MSim}}
\newcommand{\QAP}{\textsc{QAP}}
\newcommand{\maxQAP}{\textsc{max-QAP}}

\usepackage[textwidth=2cm]{todonotes}

\newcommand{\rank}{\operatorname{rank}}
\newcommand{\rows}{\operatorname{rows}}

\title{Graph Similarity and Approximate Isomorphism}
\author{Martin Grohe, Gaurav Rattan, Gerhard J. Woeginger 
\\ RWTH Aachen University, Germany 
\\ \normalsize\{grohe,rattan,woeginger\}@informatik.rwth-aachen.de}
\date{}

\begin{document}
\maketitle
\sloppy

\begin{abstract} 
The \emph{graph similarity} problem, also known as approximate graph
isomorphism or graph matching problem, has been extensively studied
in the machine learning community, but has not received much
attention in the algorithms community:
Given two graphs $G,H$ of the same order $n$ with adjacency matrices 
$A_G,A_H$, a well-studied measure of similarity is the \emph{Frobenius distance}
\[
\dist(G,H):=\min_{\pi}\|A_G^\pi-A_H\|_F,
\]
where $\pi$ ranges over all permutations of the vertex set of $G$, where 
$A_G^\pi$ denotes the matrix obtained from $A_G$ by permuting rows and 
columns according to $\pi$, and where $\|M\|_F$ is the Frobenius norm of 
a matrix $M$. The (weighted) graph similarity problem, denoted by $\SIM$ 
($\WSIM$), is the problem of computing this distance for two graphs of 
same order. This problem is closely related to the notoriously hard 
\emph{quadratic assignment problem} ($\QAP$), which is known to be 
\NP-hard even for severely restricted cases. 

It is known that $\SIM$ ($\WSIM$) is \NP-hard; we strengthen this
hardness result by showing that the problem remains \NP-hard even
for the class of trees. Identifying the boundary of tractability for
$\WSIM$ is best done in the framework of linear algebra. 
We show that $\WSIM$ is \NP-hard as long as one of
the matrices has unbounded rank or negative eigenvalues: hence,
the realm of tractability is restricted to positive 
semi-definite matrices of bounded rank. 
Our main result is a polynomial time algorithm 
for the special case where one of the matrices has 
a bounded \emph{clustering number}, a parameter arising from 
spectral graph drawing techniques.
\end{abstract}


\section{Introduction}\label{intro}
Graph isomorphism has been a central open problem in algorithmics for the last 50 years. 
The question of whether graph isomorphism is in polynomial time is still wide open, but 
at least we know that it is in quasi-polynomial time \cite{bab16}. 
On the practical side, the problem is largely viewed as solved; there are excellent tools
\cite{codkatsakmar13,junkas07,mck81,mckpip14} that efficiently decide isomorphism on all 
but very contrived graphs \cite{neuschwe17}. 
However, for many applications, notably in machine learning, we only need to know whether 
two graphs are ``approximately isomorphic'', or more generally, how ``similar'' they are. 
The resulting \emph{graph similarity} problem has been extensively studied in the machine 
learning literature under the name \emph{graph matching}
(e.g.~\cite{almduf93,confogsanven04,golran96,ume88,zasbacver09}), and also in the context 
of the schema matching problem in database systems (e.g. \cite{melgarrah02}). 
Given the practical significance of the problem, surprisingly few theoretical results are known. 
Before we discuss these known and our new results, let us state the problem formally.

\paragraph{Graph similarity.}
It is not obvious how to define the distance between two graphs, but the distance measure that
we study here seems to be the most straightforward one, and it certainly is the one that has
been studied most. 
For two $n$-vertex graphs $G$ and $H$ with adjacency matrices $A_G$ and $A_H$, we define the 
\emph{Frobenius distance} between $G$ and $H$ to be
\begin{equation}
  \label{eq:1}
  \dist(G,H):=~ \min_{\pi}\|A_G^\pi-A_H\|_F.
\end{equation}
Here $\pi$ ranges over all permutations of the vertex set of $G$,
$A_G^\pi$ denotes the matrix obtained from $A_G$ by permuting rows and
columns according to $\pi$, and the norm
$\|M\|_F:=\sqrt{\sum_{i,j}M_{ij}^2}$ is the \emph{Frobenius norm} of a
matrix $M=(M_{ij})$.  Note that $\dist(G,H)^2$ counts the number of
edge mismatches in an optimal alignment of the two graphs.  The
\emph{graph similarity problem}, denoted by $\SIM$, is the problem of
computing $\dist(G,H)$ for graphs $G,H$ of the same order, or,
depending on the context, the decision version of this problem (decide
whether $\dist(G,H)\le d$ for a given $d$).  We can easily extend the
definitions to weighted graphs and denote the \emph{weighted graph
  similarity problem} by $\WSIM$.  In practice, this is often the more
relevant problem.  Instead of the adjacency matrices of graphs, we may
also use the Laplacian matrices of the graphs to define distances.
Recall that the \emph{Laplacian matrix} of a graph $G$ is the matrix
$L_G:=D_G-A_G$, where $D_G$ is the diagonal matrix in which the entry
$(D_G)_{ii}$ is the degree of the $i$th vertex, or in the weighted
case, the sum of the weights of the incident edges.  Let
$\dist_L(G,H):=\min_{\pi}\|L_G^\pi-L_H\|_F$ be the corresponding
distance measure.  Intuitively, in the definition of $\dist_L(G,H)$ we
prefer permutations that map vertices of similar degrees onto one
another.  Technically, $\dist_L(G,H)$ is interesting, because the
Laplacian matrices are positive semidefinite (if the weights are
nonnegative).  Both the (weighted) similarity problem and its version
for the Laplacian matrices are special cases of the problem $\MSIM$ of
computing $\min_{P}\|A-PBP^{-1}\|_F$ for given symmetric matrices
$A,B\in\RR^{n\times n}$.  In the Laplacian case, these matrices are
positive semidefinite.\footnote{Note that the notion of similarity
  that we use here has nothing to do with the standard notion of
  ``matrix similarity'' from linear algebra.}

\paragraph{The QAP.}
The graph similarity problem is closely related to \emph{quadratic assignment problem ($\QAP$)} \cite{Cela-qap1998}: 
given two $(n\times n)$-matrices $A,B$, the goal is to find a permutation $\pi\in S_n$ that minimizes 
$\sum_{i,j}A_{ij}B_{\pi(i)\pi(j)}$. 
The usual interpretation is that we have $n$ \emph{facilities} that we want to assign to $n$ \emph{locations}. 
The entry $A_{ij}$ is the \emph{flow} from the $i$th to the $j$th facility, and the entry $B_{ij}$ is the
\emph{distance} from the $i$th to the $j$th location. 
The goal is to find an assignment of facilities to locations that minimizes the total cost, where the cost 
for each pair of facilities is defined as the flow times the distance between their locations. 
The $\QAP$ has a large number of real-world applications, as for instance
hospital planning \cite{Elshafei1977},
typewriter keyboard design \cite{PoGeRa1976},
ranking of archeological data \cite{KrPr1978}, and
scheduling parallel production lines \cite{GeGr1976}.
On the theoretical side, the $\QAP$ contains well-known optimization problems as special cases, 
as for instance the Travelling Salesman Problem, 
the feedback arc set problem, 
the maximum clique problem,
and all kinds of problems centered around graph partitioning, graph embedding, and graph packing.

In the maximization version $\maxQAP$ of $\QAP$ the objective is to maximize 
$\sum_{i,j}A_{ij}B_{\pi(i)\pi(j)}$ (see \cite{makmansvi14,nagsvi09}).
Both $\QAP$ and $\maxQAP$ are notoriously hard combinatorial optimization problems, in terms of
practical solvability \cite{ReWo1992} as well as in terms of theoretical hardness results even 
for very restricted special cases \cite{qap-1,qap-2,qap-3}.
It is easy to see that $\MSIM$ is equivalent to $\maxQAP$, because in reductions between $\QAP$ 
and $\MSIM$ the sign of one of the two matrices is flipped.
Most of the known results for $\SIM$ and its variants are derived from results for $\textsc{(max)QAP}$.

\paragraph{Previous Work.}
It seems to be folklore knowledge that $\SIM$ is NP-complete. 
For example, this can be seen by a reduction from the Hamiltonian path problem: take $G$ to be the 
$n$-vertex input graph and $H$ a path of length $n$; then $\dist(G,H)\le\sqrt{|E(G)|-n}$ if and
only if $G$ has a Hamiltonian path. 
By the same argument, we can actually reduce the subgraph isomorphism problem to $\SIM$.
Arvind, K\"obler, Kuhnert, and Vasudev~\cite{arvkobkuhvas12} study several versions of what they 
call \emph{approximate graph isomorphism}; their problem \textsc{Min-PGI} is the same as our $\SIM$. 
They prove various hardness of approximation results. 
Based on an earlier QAP-approximation algorithm due to Arora, Frieze, and Kaplan~\cite{arofrikap02}, 
they also obtain a quasi-polynomial time approximation algorithm for the related problem \textsc{Max-PGI}.
Further hardness results were obtained by Makarychev, Manokaran, and Sviridenko \cite{makmansvi14} 
and O'Donnell, Wright, Wu, and Zhou \cite{odowriwu+14}, who prove an average case hardness result for 
a variant of $\SIM$ problem that they call \emph{robust graph isomorphism}. 
Keldenich~\cite{master-kel} studied the similarity problem for a wide range matrix norms (instead
of the Frobenius norm) and proved hardness for essentially all of them.

\paragraph{Spectral Graph Visualization.}
Since $\WSIM$ and $\MSIM$ are essentially linear algebraic problems, 
it is reasonable to hope that the spectral structure 
of the input (adjacency) matrices is closely related with 
the computational complexity of these problems.  
In this regard, we remark that 
spectral graph drawing is a well-established technique for 
visualizing graphs via their spectral properties. 
Formally, let $G$ be a $n$-vertex graph: 
a graph drawing is a map $\rho:V(G) \mapsto \reals^k$, 
where the ambient space has dimension $k \ll n$. 
For spectral graph drawings, this map is typically defined as follows. 
We select a suitable matrix representation 
of the graph and select up to $k$ eigenvectors $u_1,\dots,u_k$ of this matrix. 
Then, the mapping $\rho:V(G) \mapsto \reals^k$ is defined by the 
rows $\{r_1,\dots ,r_n\}$ of the $n \times k$ matrix $[u_1 \cdots u_k]$. 
The choice of the matrix representation and the selection of eigenvectors 
usually depends on the problem at hand. The most useful matrix representation in 
the spectral drawing framework is the well-known Laplacian matrix:
the eigenvectors $u_1,\dots,u_k$ corresponding to $k$ smallest eigenvalues 
define the drawing $\rho$ of interest.

Observe that the graph drawing $\rho$ defined above is not injective in general.
Given such a drawing $\rho$, we define the \emph{clustering number} of a graph $G$ 
to be the cardinality of the set $\mathrm{Image}(\rho)$. 
The elements of $\mathrm{Image}(\rho)$ correspond to subsets of $V(G)$:
every vertex in such a `cluster' has identical adjacency. 

\paragraph{Our results.}
So where does all this leave us? 
Well, $\SIM$ is obviously an extremely hard optimization problem. 
We start our investigations by adding to the body of known hardness results: we prove that $\SIM$ 
remains NP-hard even if both input graphs are trees (Theorem~\ref{thm:hard:trees}). 
Note that in strong contrast to this, the subgraph isomorphism problem becomes easy if both input 
graphs are trees \cite{mat78}. 
The reduction from Hamiltonian path sketched above shows that $\SIM$ is also hard if one input graph is a path. 
We prove that $\SIM$ is tractable in the very restricted case that one of the input graphs is a path and the 
other one is a tree (Theorem~\ref{thm:path-tree}).

As $\WSIM$ and $\MSIM$ are essentially linear algebraic problems, 
it makes sense to look for algebraic tractability criteria. 
We explore bounded rank (of the adjacency matrices) as a tractability
criteria for $\WSIM$ and $\MSIM$. Indeed, the \NP-hardness reductions for $\SIM$ 
involve graphs which have adjacency matrices of high rank (e.g.~paths, cycles). 
We show that the problem $\SIM$ (and $\WSIM$) remains \NP-hard 
as long as one of the matrices has unbounded rank or negative eigenvalues.
(Theorems~\ref{thm:hard:matrices}, \ref{thm:hard:psd} and \ref{thm:hard:psd:bnd}).
Consequently, the realm of tractability for $\WSIM$ (and $\MSIM$) is 
restricted to the class of positive semi-definite matrices of bounded rank. 
We feel that for a problem as hard as $\QAP$ or $\MSIM$, identifying any somewhat natural 
tractable special case is worthwhile. Our main result (Theorem~\ref{thm:main}) is a polynomial time algorithm for $\MSIM$ if both input matrices 
are positive semidefinite (as it is the case for the Laplacian version of $\WSIM$) and have bounded-rank, 
and where one of the matrices has a bounded clustering number. 

For the proof of Theorem~\ref{thm:main}, we can re-write the (squared) objective function as $\|AP-PB\|^2_F$, where $P$ ranges 
over all permutation matrices. This is a convex function, and it would be feasible to minimize it over a convex domain. 
The real difficulty of the problem lies in the fact that we are optimizing over the complicated discrete 
space of permutation matrices. Our approach relies on a linearization of the solution space, and the key insight (Lemma~\ref{lem:lin-sep}) 
is that the optimal solution is essentially determined by polynomially many hyperplanes. 
To prove this, we exploit the convexity of the objective function in a peculiar way.

\section{Preliminaries}\label{sec:prelims}

\subsection{Notation}

We denote the set $\{1,\dots,n\}$ by $[n]$. 
Unless specified otherwise, we will always assume
that the vertex set of an $n$-vertex graph $G$ is $[n]$.
We denote the degree of a vertex $v$ by $d_G(v)$. 

\paragraph{Matrices.} 
Given an $m \times n$ matrix $M$, the $i^{th}$ row (column) of $M$ is denoted by $M^i$ ($M_i$). 
The multiset $\{M^1,\dots,M^m\}$ is denoted by $\rows(M)$. 
Given $S\subseteq [m]$, the sum $\sum_{i\in S} M^i$ is denoted by $M^S$.
We denote the $n \times n$ identity matrix by $I_n$.

A real symmetric $n \times n$ matrix $M$ is called \textit{positive semi-definite} (p.s.d), 
denoted by $M \succeq 0$, if the scalar $z^TMz$ is non-negative for every $z \in \reals^n$. 
The following conditions are well-known to be equivalent.
\begin{compactenum}
\item $M \succeq 0$
\item Every eigenvalue of $M$ is non-negative.
\item $M= W^TW$ for some $n \times n$ matrix $W$. In other words, there 
exist $n$ vectors $w_1,\dots,w_n \in \reals^n$ such that $M_{ij}= w_i^Tw_j$. 
\end{compactenum}

Given two vectors $x,y \in \reals^n$, their dot product $\innerpd{x}{}{y}$ is defined to be $x^Ty$. 
Given $M \succeq 0$, the inner product of $x,y$ w.r.t. M, denoted by $\innerpd{x}{M}{y}$, is defined to be $x^TMy$.
The usual dot product corresponds to the case $M =I$, the identity matrix. 

Every $n \times n$ symmetric matrix $M$ of rank $k$ has a spectral decomposition $M = U \Sigma U^T$.
Here, $\Sigma$ is a $k \times k$ diagonal matrix with the eigenvalues $\lambda_1,\dots,\lambda_k \in \reals$ on the diagonal. 
The matrix $U$ is a $n \times k$ matrix with the corresponding eigenvectors $v_1,\dots,v_k$ as the columns $U_1,\dots,U_k$.

\paragraph{Graphs and Matrices.}
The Laplacian matrix of a (weighted) undirected graph $G$, denoted by $L_G$, is defined as follows. 
Let $A \in \reals^{n \times n}$ be the symmetric (weighted) adjacency matrix of $G$.
Let $D$ be a $n \times n$ diagonal matrix, such that
$D_{ii}$ is the sum of weights of the edges 
incident on the $i^{th}$ vertex. For simple undirectred graphs, $D_{ii}=d_G(v_i)$.
Define the Laplacian of $G$ as $L(G) = D - A$. 
This definition allows us to express the quadratic form 
\[
 x^T L_G x = \displaystyle\sum_{\{i,j\} \in E(G)} a_{ij} (x_i-x_j)^2. 
\]
The above expression immediately implies that $L_G$ is positive semi-definite. 

\paragraph{Clustering Number.}
Recall the following definitions from Section \ref{intro}.
Given a $n$-vertex graph $G$, a \emph{graph drawing} is a 
map $\rho:V(G) \mapsto \reals^k$, where the ambient dimension $k \ll n$. 
We will use the adjacency matrix $A$ of a graph $G$ 
to generate spectral graph drawings as follows. Let the rank of $A$ be $k$, 
and let $\specA$ be a spectral decomposition. Denote $U = [u_1\cdots u_k]$,
where $u_1,\dots,u_k$ are the eigenvectors of $A$. 
The mapping of our interest $\rho:V(G) \mapsto \reals^k$
is defined by the rows $\{r_1,\dots ,r_n\}$ of the $n \times k$ matrix $U$.
Given any two spectral decompositions $\specA$ and $\specprA$, 
it holds that $U' = U O_k$ for some $k \times k$ orthogonal matrix $O_k$. 
Since $O_k$ is invertible, the number of distinct tuples in the set $\rows(U)$ 
is equal to the corresponding number for the set $\rows(U')$.
This allows us to define the \emph{clustering number of a graph $G$}:
it is equal to the cardinality of the set $\mathrm{Image}(\rho)$,
where $\rho$ is defined via some spectral decomposition of $A$, as above.
The above definitions generalize to weighted (undirected) graphs in an analogous manner. 

\paragraph{Frobenius Norm.}
The trace of a matrix $M$, denoted by $\trace{M}$, is defined to be $\sum_{i \in [n]} M_{ii}$. 
The \emph{trace inner product} of two matrices $A$ and $B$, denoted by $\tracepd{A}{B}$, is the scalar $\trace{A^TB}$.  
The \emph{Frobenius norm} $\frob{M}$ of a matrix $M$ is defined in the
introduction. It is easy to check that $\frob{M}^2 = \tracepd{M}{M}$. 

Given two $n$-vertex graphs $G$ and $H$ and a permutation $\pi \in S_n$, 
a $\pi$-\emph{mismatch} between $G$ and $H$ is a pair $\{i,j\}$ such that $\{i,j\} \in E(G)$ and $\{i^{\pi},j^{\pi}\} \notin E(H)$ (or vice-versa). 
In other words, $\pi:V(G) \rightarrow V(H)$ does not preserve adjacency for the pair $\{i,j\}$. The following claim will be useful 
as a combinatorial interpretation of the Frobenius norm. 
Let $\Delta$ denote the number of $\pi$-mismatches between $G$ and $H$. 
\begin{claim}
 $\|A_G^{\pi} - A_H\|_F^2 = 2 \Delta$.
\end{claim}
\begin{proof}
The only non-zero terms in the expansion of summation $\|A_G^{\pi} - A_H\|_F^2$ correspond to $\pi$-mismatches. 
Since every mismatch $\{i,j\}$ contributes $1$ and is counted twice in the summation, the claim follows.
\end{proof}

\subsection{Convex Optimization}\label{subsec:convex}
A \emph{hyperplane} $H$ in the Euclidean space $\reals^k$ is a $(k-1)$-dimensional affine subspace.
The usual representation of a hyperplane is a linear equation $\innerpd{c}{}{x}=\alpha$ for some $c\in\reals^k, \alpha \in \reals$. 
The convex sets $\{x \,\vert\, \innerpd{c}{}{x} > \alpha\}$ and $\{x \,\vert\, \innerpd{c}{}{x} < \alpha\}$ are called the open \emph{half-spaces} corresponding to $H$, 
denoted by $H^+,H^-$ respectively. 

Two sets $(S,T)$ are \emph{weakly linearly separated} if there exists a hyperplane $H$ 
such that $S \subseteq H^+ \cup H$ and $T\subseteq H^- \cup H$. In this case, we call them to be 
weakly linearly separated along $H$. A family of sets $S_1,\dots,S_p$ is \emph{weakly linearly separated} 
if for every $l,m \in [p]$, the sets $S_l,S_m$ are weakly linearly separated. 
Let $\Pi$ be a partition of a set $S$ into $p$ sets $S_1,\dots,S_p$. The
partition $\Pi$ is said to be \emph{mutually linearly separated} 
if the family of sets $S_1,\dots,S_p$ is weakly linearly separated. 

%
%

 A subset $S \subseteq \reals^k$ is called \emph{convex} if for every $x,y \in S$, $\alpha x + (1-\alpha)y \in S$, $\alpha \in [0,1]$. A function $f:\reals^k \rightarrow \reals$ is called \emph{convex} on a convex set $S$ if for every $x,y \in S$, $f(\alpha x + (1-\alpha)y) \leq \alpha f(x) + (1-\alpha) f(y)$. 
 The following theorem about linearization of convex differentiable
 functions is well-known and is stated without proof. 
 The gradient of a function $f:\reals^k \rightarrow \reals$, denoted by $\nabla f$, is the vector-valued function $[ \frac{\partial f }{\partial x_1 } \dots \frac{\partial f }{\partial x_k }]$.
 Given $X^* \in \reals^k$, let $\mu^*$ denote the vector $\nabla f (X^*)$.

\begin{theorem}[Convex function linearization]\label{thm:convlin}
Let $f: \reals^k \rightarrow \reals$ be a convex function. 
For all $X\in \reals^k$, $f(X) - f(X^*) \geq \innerpd{\mu^*}{}{X-X^*}$.
\end{theorem}

Next, we show that the linearization of a convex function can be useful
in understanding its optima over a finite domain. We prove the following lemma about 
convex functions, which is interesting in its own right. 

\begin{lemma}
  Let $\Omega$ be a finite subset of $\reals^{k} \times \reals^ {\ell}$. Let $G:\reals^{k} \rightarrow \reals$, 
$H:\reals^{\ell}\rightarrow \reals$ such that $H$ is convex, and let $F:\reals^{k} \times
\reals^ {\ell} \rightarrow \reals$ be defined as $F(X,Y) = G(X) +
H(Y)$. Let $(X^*,Y^*) \in \argmax_{(X,Y)\in \Omega} F(X,Y)$.

Then there exist a $\mu^* \in \reals^{\ell}$ such that:
\begin{enumerate}
\item[(i)] $(X^*,Y^*) \in \argmax_{(X,Y)\in \Omega} L(X,Y)$ where
  $L(X,Y) = G(X) + \innerpd{\mu^*}{}{Y}$;
\item[(ii)] $\argmax_{(X,Y) \in \Omega} L(X,Y)\subseteq\argmax_{(X,Y) \in \Omega} F(X,Y)$. 
\end{enumerate}
\end{lemma}



In other words, for every $(X^*,Y^*)$ which maximizes $F$ over $\Omega$, there exists a partially ``linearized'' function $L$ such that $(X^*,Y^*)$ maximizes $L$ over $\Omega$. Moreover, every maximizer of $L$ over $\Omega$ is a maximizer of $F$ over $\Omega$. This additional condition is necessary so that this ``linearization'' does not create spurious optimal solutions. 

\begin{proof} 
  Let $(X^*,Y^*) \in \argmax_{S \in \Omega} F(S)$.  Since $H$ is
  convex, we can use Theorem \ref{thm:convlin} to linearize $H$ around
  $Y^* \in \reals^{\ell}$. Hence, there exists a
  $\mu^* \in \reals^{\ell}$ such that
  $H(Y) - H(Y^*) \geq \innerpd{\mu^*}{}{Y -Y^*}$, or equivalently,
  \begin{equation}
    \label{eq:2}
    H(Y)-\innerpd{\mu^*}{}{Y}\geq H(Y^*)-\innerpd{\mu^*}{}{Y^*},
  \end{equation}
  for all
  $Y \in \reals^{\ell}$.
  Hence with $L(X,Y)=G(X)+\innerpd{\mu^*}{}{Y}$, for all $(X,Y)\in\Omega$ we have
  \[
    L(X^*,Y^*)=F(X^*,Y^*)-H(Y^*)+\innerpd{\mu^*}{}{Y^*}\geq
    F(X,Y)-H(Y)+\innerpd{\mu^*}{}{Y}=L(X,Y),
  \]
where the inequality holds by \eqref{eq:2} and because $(X^*,Y^*)$
maximizes $F$. Hence $(X^*,Y^*)$ maximizes $L$ as well, which proves (i).

For (ii), consider $(X^{**},Y^{**})\in \argmax_{(X,Y)\in \Omega} L(X,Y)$. 
To prove that $(X^{**},Y^{**})\in \argmax_{(X,Y)\in \Omega} F(X,Y)$, it 
suffices to prove that $F(X^{**},Y^{**})\ge F(X^*,Y^*)$. 
By (i), we have $L(X^*,Y^*)=L(X^{**},Y^{**})$. Thus 
\[
F(X^{**},Y^{**})=L(X^{**},Y^{**})+H(Y^{**})-\innerpd{\mu^*}{}{Y^{**}}\ge
L(X^{*},Y^{*})+H(Y^{*})-\innerpd{\mu^*}{}{Y^{*}}=F(X^*,Y^*),
\]
where the inequality holds by \eqref{eq:2} with $(X,Y):=(X^{**},Y^{**})$ and as
$(X^{**},Y^{**})$ maximizes $L$.
\end{proof}

 \begin{corollary}\label{coro:conv-lin}
   Let $\Omega$ be a finite subset of $\reals^{kp}$. For all
   $i\in[k]$, let $G_i:\reals^k\to \reals$ be a convex function, and
   let $F: \reals^{kp}\to\reals$ be defined ny
   $F(X_1,\ldots,X_k):=G_1(X_1)+\ldots+G_k(X_k)$. 
   Let $X^*=(X_1^*,\ldots,X_k^*)\in\argmax_{X\in\Omega} F(X)$.

Then there are
   $\mu_1^*,\ldots,\mu_k^*\in\RR^p$ such that:
\begin{enumerate}
\item[(i)] $X^* \in \argmax_{X\in \Omega} L(X)$ where
  $L(X_1,\ldots,X_k) =\sum_{i=1}^k\innerpd{\mu_i^*}{}{X_i}$;
\item[(ii)] $\argmax_{X \in \Omega} L(X)\subseteq\argmax_{ \in \Omega} F(X)$. 
\end{enumerate}
 \end{corollary}

 \begin{proof}
   Inductively apply the lemma to the functions 
   \[
     F^i((X_1,\ldots,X_{i-1},X_{i+1},\ldots,X_k),X_i)=
     \underbrace{\left(\sum_{j=1}^{i-1}\innerpd{\mu_j^*}{}{X_j}+\sum_{j=i+1}^{k}G_j(X_j)\right)}_{=:G^i(X_1,\ldots,X_{i-1},X_{i+1},\ldots,X_k)}+
     \underbrace{G_i(X_i)}_{=:H^i(X_i)}.
   \]
 \end{proof}

Finally, we state an important fact about the convexity of quadratic functions.  
Given a p.s.d. matrix $M \in \realmat{k}$, the quadratic function $Q_M:\reals^k \rightarrow \reals$ is defined as $Q_M(x) = \innerpd{x}{M}{x}$.
\begin{lemma}[Convexity of p.s.d]\label{lem:psdconv}
  $Q_M$ is convex on $\reals^k$. 
\end{lemma}
\begin{proof}
For all $\alpha \in [0,1]$, $Q_M(\alpha x + (1-\alpha) y ) = 
\innerpd{(\alpha x + (1-\alpha y)}{M}{(\alpha x + (1-\alpha y)} = \alpha^2 \innerpd{x}{M}{x} + (1-\alpha)^2 \innerpd{y}{M}{y} + 2 \alpha(1-\alpha) \innerpd{x}{M}{y}$.
Using $\innerpd{x-y}{M}{x-y} \geq 0$, we can show that 
$\innerpd{x}{M}{x} + \innerpd{y}{M}{y} \geq 2 \innerpd{x}{M}{y}$. Combining, we have $Q_M(\alpha x + (1-\alpha) y ) \leq \alpha^2 Q_M(x) + (1-\alpha)^2 Q_M(y) + \alpha(1-\alpha) (Q_M(x) + Q_M (y)) \leq \alpha Q_M(x) + (1-\alpha) Q_M(y)$. 
Hence, $Q_M$ is convex. 
\end{proof}

\subsection{Simulation of Simplicity} 
\label{subsec:sos}
In this section, we describe an elegant technique for handling 
degeneracy in the input data for geometrical algorithms that is 
due to Edelsbrunner and M\"ucke \cite{edel90}. 
We also state an important lemma which will be directly useful 
for our algorithmic results in Section \ref{sec:algo}. 

An input set $S$ of $n$ points $w_1,\dots,w_n \in \reals^k$ is
said to be in \emph{general position}, if there is no subset
$S' \subseteq S$ with $|S'|>k$ that lies on a common hyperplane.
If we are optimizing a certain function of this input on a discrete 
space $\Omega$, infinitesimally small perturbations of 
$w_1,\dots,w_n$ will not change the set $\Omega^* \subseteq \Omega$ 
of optimal solutions.
Hence we may always assume (modulo infinitesimal perturbations)
that such input sets are in general position and do not contain
degenerate subsets.
{From} the algorithmic point of view, the caveat is that these 
perturbations might be so small that we cannot even represent 
them efficiently. 

In this context, Edelsbrunner and M\"ucke \cite{edel90} developed 
a useful technique to handle degeneracy in input data, called 
\emph{Simulation-of-Simplicity}. 
The idea is to introduce conceptual perturbations which eliminate 
all degeneracies: the perturbations are never computed explicitly
in practice. 
In fact, the perturbations are just certain conveniently chosen 
polynomials in a parameter $\epsilon$, so that after adding
these polynomials to the coordinates the perturbed set agrees with 
the input set for $\epsilon=0$. 
For our purposes, we require such a perturbation of an input set $S$ 
of $n$ points $w_1,\dots,w_n \in \reals^k$ that brings them into 
general position. 
We select $nk$ perturbations $\epsilon_{ij}$ for $i\in[k]$ and $j \in [n]$ as follows. 
We perturb the $i^{th}$ coordinate of vector $w_j$ by adding $\epsilon_{ij}$. 
In our algorithmic application, we need to consistently answer 
queries of the type: \emph{``Given points $w_{i_1},\dots,w_{i_k}$ (with 
$i_1 < \dots <i_k$) and a point $w_{i_{k+1}}$, does the point $w_{i_{k+1}}$ 
lie below, on, or above the hyperplane determined by $w_{i_1},\dots,w_{i_k}$?''}
We can implement and answer such queries in $\mathcal{O}(k^k)$ time as 
follows. 
The answer to the query depends on the sign of the determinant of the following 
$(k+1) \times (k+1)$ matrix $\tilde{M}$, which is also the \emph{signed} 
volume of the parallelopiped defined by the vectors $w_{i_1}-w_{i_{k+1}},\dots,w_{i_k}-w_{i_{k+1}}$.
\begin{align*}
\left[
\begin{array}{c c c c c }
  w_{i_11} + \epsilon_{i_1 1} & w_{i_21}+\epsilon_{i_2 1} & \dots & w_{i_k 1}+\epsilon_{i_k 1} & w_{i_{k+1}1}+\epsilon_{i_{k+1} 1} \\
  \vdots & \dots & \vdots & \vdots & \vdots \\
  w_{i_1k} + \epsilon_{i_1 k} & w_{i_2k}+\epsilon_{i_2 k} & \dots & w_{i_k k}+\epsilon_{i_k k} & w_{i_{k+1} k}+\epsilon_{i_{k+1} k} \\
 1 & 1 & \dots & 1 & 1 
\end{array}
\right].
\end{align*}  
The determinant of matrix $\tilde{M}$ is a polynomial in the $\epsilon_{ij}$,
which can be computed in $\bigO((k+1)!)= \mathcal{O}(k^k)$ time by using the
Leibniz expansion
\[
\operatorname{det}(\tilde{M}) = 
  \displaystyle\sum_{\sigma \in S_{k+1}}(\operatorname{sgn}(\sigma)\displaystyle\prod_{i=1}^{k+1} \tilde{M}_{i,i^{\sigma}}).
\]
It is easy to see that this polynomial is not identically zero, as every 
term in the Leibniz expansion yields a different polynomial. 
This property ensures the non-degeneracy in our conceptual perturbations. 
We impose a lexicographic ordering on $\epsilon_{ij}$ as follows: 
$\epsilon_{11} < \dots <\epsilon_{1n}< \epsilon_{21}<\dots<\epsilon_{2n}<\dots<\epsilon_{kn}$. 
This induces a natural lexicographic ordering on the monomials in the polynomial $\operatorname{det}(\tilde{M})$.
The lexicographically least monomial in this ordering has either a positive or a negative coefficient:
we interpret the sign of this coefficient as the relative position of $w_{i_{k+1}}$ with respect 
to the hyperplane determined by $w_{i_1},\dots,w_{i_k}$. We refer the reader to \cite{edel90} for 
further details. We summarize the above discussion in the following lemma. 

\begin{lemma}\label{lem:sos}
Given a set $W=\{w_1,\dots,w_n\}$ of $n$ points in $\reals^k$, 
\begin{itemize}
\item The lexicographic ordering of the $\epsilon_{ij}$ yields a canonical perturbation of the 
points $w_1,\dots,w_n$ such that the resulting set is in general position. 
\item There exists an $\mathcal{O}(k^k)$ time subroutine which computes the relative position of 
a canonically perturbed point with respect to the hyperplane determined by $k$ canonically 
perturbed points.
\end{itemize}
\end{lemma}


\section{Hardness Results}
\label{sec:hard}
In this section, we show several new hardness results for problems $\SIM, \WSIM$ and $\MSIM$.
As we will observe, these problems turn out to be algorithmically intractable, even for severely restricted cases. 
We begin by recalling the following observation. 

\begin{theorem}[Folklore]
\label{thm:hard:graphs}
 $\SIM$ is \NP-hard for the class of simple undirected graphs. 
\end{theorem}

In fact, the problem turns out to be \NP-hard even for very restricted graph classes. 
The following theorem is the main hardness result of this section.  
\begin{theorem}
\label{thm:hard:trees}
 $\SIM$ is \NP-hard for the class of trees. 
\end{theorem}

On the other hand, if we restrict one of the input instances to be a path, 
the problem can be solved in polynomial time. 
The following theorem provides a positive example of tractability of $\SIM$. 
\begin{theorem}
\label{thm:path-tree}
An input instance $(G,H)$ of $\SIM$, where $G$ is a path and $H$ is a tree, can be solved in polynomial time. 
\end{theorem}

The above results exhibit the hardness of $\SIM$, and consequently, the hardness of the more general 
problems $\WSIM$ and $\MSIM$. 
Since the graphs (for instance cycles and paths) involved in the hardness reductions have adjacency 
matrices of high rank, it is natural to ask whether $\MSIM$ would become tractable for matrices of low rank. 
Our following theorem shows that $\MSIM$ is \NP-hard even for matrices of rank at most $2$. 
The underlying reason for hardness is the well-known problem \qap, which shares the optimization domain $\Sn$. 

\begin{theorem}
\label{thm:hard:matrices}
 $\MSIM$ is \NP-hard for symmetric matrices of rank at most $2$.
\end{theorem}

The key to the above reduction is the fact that one of the matrices has non-negative Eigenvalues 
while the other matrix has non-positive Eigenvalues.
We show that the $\MSIM$ is \NP-hard even for positive semi-definite matrices. 
The main idea is to reformulate the hardness reduction in Theorem \ref{thm:hard:graphs} in terms of Laplacian matrices. 

\begin{theorem}
\label{thm:hard:psd}
 $\MSIM$ is \NP-hard for positive semi-definite matrices. 
\end{theorem}

In fact, we show that the problem remains \NP-hard, even if one of the matrices is of rank $1$. 
The proof follows by modifying the matrices in the proof of Theorem \ref{thm:hard:matrices}
so that they are positive semi-definite.
\begin{theorem}
\label{thm:hard:psd:bnd}
 $\MSIM$ is \NP-hard for positive semi-definite matrices, even if one of the matrices has rank $1$. 
\end{theorem}

Therefore, the realm of tractability for $\MSIM$ is restricted 
to positive definite matrices of low rank. In the next section, 
we prove algorithmic results in this direction. 

\section{Algorithmic Results}
\label{sec:algo}
In this section, we present the main algorithmic result of this paper. 
As established in the previous section, the domain of tractability for $\MSIM$ 
is restricted to p.s.d. matrices with low rank. The main theorem of this section 
is stated as follows. Given an instance $(A,B)$ of $\MSIM$, let $\rank(A),\rank(B)$ $\leq k$. Let $p$ be the clustering number of $B$. 

\begin{theorem}\label{thm:main}
There is a \rtime algorithm for $\MSIM$. Here, the $\bigOstar$ notation hides factors polynomial in the size of input representation. 
\end{theorem}

In order to prove Theorem \ref{thm:main}, we define a closely related optimization problem, called the \textsc{Quadratic-Vector-Partition} (\qvp).
Let $\mathcal{P}$ be the set of all (ordered) partitions of $[n]$ into $p$ sets
of size $n_1,\dots,n_p$. I.e., an element $P \in \mathcal{P}$ 
is an ordered partition $T_1 \cup \dots \cup T_p$ of $[n]$, where $|T_l| = n_l, l \in [p]$.
Given a set $W$ of $n$ vectors $\{w_1,\dots,w_n\} \subseteq \reals^k$,
we will employ two important notations. Denote $W[T_i]$ to be the point-set $\{w_j \,\vert\, j \in T_i\}$ corresponding to $T_i \subseteq [n]$. 
Denote $W^T = \sum_{i\in T} w_i$, $T \subseteq [n]$.  

The input instance to \qvp is a set $W$ of $n$ vectors $\{w_1,\dots,w_n\} \subseteq \reals^k$, along with two matrices $K$ and $\Lambda$. The matrix $K$ is a p.s.d matrix of size $p \times p$. The matrix $\Lambda$ is a \emph{diagonal} matrix with $k$ positive entries. The objective is to search for a partition $P \in \mathcal{P}$ which maximizes the following quadratic objective function $F$.
\[
 F(P) = \displaystyle\sum_{l,m \in [p]} K_{lm} \left\langle W^{T_l},W^{T_m}\right\rangle_{\Lambda}.
\]
Informally, the goal is to `cluster' the set $W$ into $p$ sets $W_1,\dots,W_p$ of cardinalities $n_1,\dots,n_p$ such that the quadratic function above is maximized. 
The connection to $\MSIM$ arises due to the following observation. We can interpret a permutation $\pi$ as a bijection $\pi:\rows(U) \rightarrow \rows(V)$ where $\specA$ and $\specB$ are the respective spectral decompositions. Since $\rank(A),\rank(B) \leq k$, we must have $U,V \in \reals^{n \times k}$ and consequently, $\rows(U),\rows(V) \subseteq \reals^k$. Since 
the set $\rows(V)$ has only $p$ distinct tuples (the clustering number), it suffices to examine the partitions of $\rows(U)$ into $p$ sets of certain fixed cardinalities. It remains then to show that the minimization of the objective function for $\MSIM$ can be reformulated as the maximization of the objective function for $\qvp$.

The proof of Theorem \ref{thm:main} proceeds in three steps. 
First, in Section \ref{subsec:reduce}, we show a reduction from $\MSIM$ to \qvp. In particular, the dimension $k$ and the parameter $p$ for the \qvp instance are equal to the rank $k$ and the clustering number $p$ in Theorem \ref{thm:main} 
respectively. Second, in Section \ref{subsec:optimal}, we show that the optimal solutions for a \qvp instance have a nice geometrical structure. In particular, the convex-hulls of the point-sets in the partition are mutually disjoint (upto some caveats). Third, in Section \ref{subsec:algo}, we describe a \rtime algorithm for \qvp. The algorithm essentially enumerates all partitions with the optimal solution structure. This finishes the proof of Theorem \ref{thm:main}. 

\subsection{Reduction to \qvp}\label{subsec:reduce}

In this subsection, we prove the following reduction lemma. 
Given two matrices $A,B \in \realmat{n}$, let $\rank(A), \rank(B)$ $\leq k$. 
Let $p$ be the cluster-number of $B$.
\begin{lemma}\label{lem:reduce}
 Given a $\MSIM$ instance $(A,B)$, we can compute a \qvp-instance
 $W,K,\Lambda$,  where $W \subseteq \reals^k$ of size $n$ and $K\in \reals^{p\times p}$, $\Lambda \in \reals^{k \times k}$, in $\bigOstar(1)$ time such that the following holds.
Given an optimal solution for the \qvp-instance $W \subseteq \reals^k$, 
we can compute $\min_{\pi \in S_n} \frob{A^{\pi} - B}$ in $O(1)$ time. 
\end{lemma}
Therefore, it suffices to design a \rtime algorithm for \qvp for the proof of Theorem \ref{thm:main}. 
The proof of Lemma~\ref{lem:reduce} is deferred to the appendix.

\subsection{Optimal Structure of \qvp}\label{subsec:optimal}

In this section, we show that the optimal solutions for a \qvp
instance have, in fact, a nice geometrical structure.  Let
$\Omega^* \subseteq \mathcal{P}$ denote the set of optimal solutions
for a \qvp instance $W,K,\Lambda$,
where $W \subseteq \reals^k$ of size $n$. Recall from Section \ref{sec:prelims} 
  that a partition $W_1,\dots,W_p$ of $W$ is \emph{mutually linearly
    separated} if for every $i,j \in [n]$, there exists a hyperplane
  $H_{ij}$ which weakly linearly separates $W_i$ and $W_j$.

\begin{lemma}\label{lem:lin-sep}
Let $P=(T_1,\dots,T_p) \in \Omega^*$ be an optimal partition for a \qvp instance $W,K,\Lambda$.
The corresponding partition $W[T_1],\dots,W[T_p]$ is mutually linearly separated. 
\end{lemma}

The proof of Lemma \ref{lem:lin-sep} proceeds in three steps. Claim \ref{clm:lin-sep1} 
shows that we can reformulate \qvp as a convex programming problem in $\reals^p$.
Claim \ref{clm:lin-sep2} stipulates certain necessary conditions for optimality, in this 
reformulated version. Using Claim \ref{clm:lin-sep3}, we revert back to the original \qvp formulation in $\reals^k$. 
This allows us to interpret the optimality conditions in Claim \ref{clm:lin-sep2} as the mutually linearly separated property in Lemma \ref{lem:lin-sep}.

Given a partition $T_1,\dots,T_p$ of $W$, let $X_q$ be 
the vector of length $p$ corresponding to the $q^{th}$ coordinates of vectors $W^{T_1},\dots,W^{T_p}$. 
Formally, $X_q$ denotes the vector $[ (W^{T_1})_q \dots (W^{T_p})_q] \in \reals^p$, $q \in [k]$. 
Recall that $\Lambda$ is a diagonal matrix with positive entries, say $\lambda_1,\dots,\lambda_k$.
The following claim shows that we can describe our problem as a convex programming problem in $\reals^p$. 
The objective function is a sum of $k$ vector norms (squared). 
\begin{claim}\label{clm:lin-sep1} 
 $\Omega^* = \argmax_{P \in \mathcal{P}}  \displaystyle\sum_{q=1}^{k} \lambda_q \innerpd{X_q}{K}{X_q}$.
\end{claim}
The proof is deferred to the appendix.

The second step constitutes the key insight to the proof of Lemma \ref{lem:lin-sep}. We show that an optimal solution for the convex program of Claim \ref{clm:lin-sep1} must be an optimal solution for some \emph{linear} program. The proof of this claim builds on the statements in Subsection \ref{subsec:convex} about linearization of convex objective functions. 
Recall that the set $\Omega^* \subseteq \mathcal{P}$ denote the set of optimal solutions
for the \qvp instance $W,K,\Lambda$.
\begin{claim} \label{clm:lin-sep2} 
For every $P^* \in \Omega^*$, there exist vectors $\mu_1^*,\dots,\mu_k^* \in \reals^p$ such that 
$P^*$ is an optimal solution for the objective function
\[
L = \argmax_{P \in \mathcal{P}} \displaystyle\sum_{q=1}^{k} \lambda_q \innerpd{\mu_q^*}{K}{X_q}. 
\]
Moreover, the set of optimal solutions of $L$ is a subset of
$\Omega^*$.
\end{claim}
The proof is deferred to the appendix.

Finally, we undo the transformation of Claim \ref{clm:lin-sep1} and revert back to $\reals^k$ in 
the following claim. Consequently, we can reformulate the optimality conditions of Claim \ref{clm:lin-sep2} as follows.
\begin{claim}\label{clm:lin-sep3} 
For every $P^* \in \Omega^*$, there exist vectors $\mu_1, \dots, \mu_k \in \reals^k$ such that 
$P^*$ is an optimal solution for the objective function
\[
\mathcal{L}_{\mu_1,\dots,\mu_p} =  \max_{P \in \mathcal{P}} \displaystyle\sum_{q=1}^{p} \innerpd{\mu_q}{}{W^{T_q}} 
\]
Moreover, the set of optimal solutions of $\mathcal{L}_{\mu_1,\dots,\mu_p}$ is a subset of $\Omega^*$.
\end{claim}
The proof is deferred to the appendix.

We finish with the proof of Lemma \ref{lem:lin-sep}. 

\begin{proofof}{Lemma \ref{lem:lin-sep}}
Since $P=(T_1,\dots,T_p) \in \Omega^*$ is an optimal partition for a \qvp instance $W$,
by Claim \ref{clm:lin-sep3}, there exist vectors $\mu_1, \dots, \mu_k \in \reals^k$ such that 
$P^*$ is an optimal solution for the objective function
\[
\mathcal{L}_{\mu_1,\dots,\mu_p} =  \max_{P \in \mathcal{P}} \displaystyle\sum_{q=1}^{p} \innerpd{\mu_q}{}{W^{T_q}}. 
\]
Recall the notation $W[T_q] = \{ w_i \,\vert\, i \in T_q \}$. Suppose
there exist $q,r$ such that $W[T_q]$ and $W[T_r]$ are not (weakly) linearly
separated. We claim that this is a contradiction.  Indeed, we can
isolate the terms $\innerpd{\mu_q}{}{W^{T_q}} + \innerpd{\mu_q}{}{W^{T_r}}$ and rewrite
them as $\innerpd{(\mu_q - \mu_r)}{}{W^{T_q}} + \innerpd{\mu_r}{}{(W^{T_q} +
  W^{T_r} )}$.  Now we (weakly) linearly separate the set $W[T_q] \cup
W[T_r]$ along the direction $(\mu_q - \mu_r)$, that is, we choose a
partition $T_q'\cup T_r'$ of $T_q\cup T_r$ such that $T_q'=\{j\in
T_q\cup T_p\mid \innerpd{\mu_q-\mu_r}{}{w_j}\geq0\}$. Then
$\innerpd{(\mu_q - \mu_r)}{}{W^{T'_q}}>\innerpd{(\mu_q -
  \mu_r)}{}{W^{T_q}}$, because $T_q$ and $T_r$ are not (weakly) linearly
separated by $\mu_q-\mu_r$, and $\innerpd{\mu_r}{}{(W^{T'_q} +
  W^{T'_r} )}=\innerpd{\mu_r}{}{(W^{T_q} +
  W^{T_r} )}$, because $T_q'\cup T_r'=T_q\cup T_r$. Hence
$\innerpd{\mu_q}{}{W^{T'_q}} +
\innerpd{\mu_q}{}{W^{T'_r}}>\innerpd{\mu_q}{}{W^{T_q}} +
\innerpd{\mu_q}{}{W^{T_r}}$, which contradicts the maximality of $P^*=(T_1,\ldots,T_p)$.
Therefore, it must be the case that
the sets $T_q$ and $T_r$ are already (weakly) linearly separated along
$(\mu_q - \mu_r)$.
\end{proofof}

\subsection{Algorithm for \qvp}\label{subsec:algo}


In this subsection, we describe a \rtime algorithm for \qvp. Along with the reduction 
 stated in Lemma \ref{lem:reduce}, this finishes the proof of Theorem \ref{thm:main}.

We proceed with an informal description of the algorithm. 
 Recall that a \qvp instance is $(W,K,\Lambda)$ where $W= \{w_1,\dots,w_n\} \subset \reals^k$. 
 The output is an ordered partition $(T_1,\dots,T_p)$ of $[n]$ satisfying $|T_i| = n_i$,
 for some fixed $n_1,\dots,n_p$. Our strategy is simple: we enumerate all partitions
 $(T_1,\dots,T_p)$ of $[n]$ such that the sets $W[T_i],W[T_j]$ are weakly linearly separated 
 for every $i,j \in [p]$. By Lemma \ref{lem:lin-sep}, this suffices to obtain an optimal partition. 
 We briefly describe our algorithm.  We first guess the ${p \choose 2}$ separating hyperplanes $H_{ij}$, $i,j \in [p]$, where
 $H_{ij}$ weakly linearly separates $W[T_i]$ and $W[T_j]$. 
 Let $\mathcal{H}$ be the set of ${n \choose k}$  hyperplanes defined by $k$-subsets of $W$.
 It is sufficient to pick $H_{ij}$ from the set $\mathcal{H}$, 
 since a hyperplane in $\reals^k$ can be equivalently replaced by a hyperplane in $\mathcal{H}$,
 without changing the underlying (weakly) linear separation. These hyperplanes partition 
 $\reals^k$ into convex regions. For every $w_i \in W$, we check its relative position with respect to these hyperplanes.
 We assign $w_i$ to one of the sets $T_1,\dots,T_p$, depending of its relative position.
  We claim 
 that every weakly linearly separated family of sets $W[T_1],\dots,W[T_p]$ can be discovered on some branch of our 
 computation. The choice of $p^2$ hyperplanes
 implies a ${n \choose k}^{p^2}$ branching. 
 Therefore, the overall branching factor is $n^{\mathcal{O}(kp^2)}$. Algorithm \ref{alg:main}
 gives a formal description of our algorithm.

There are two caveats. First, we also pick an orientation $\sigma_{ij} \in \{+1,-1\}$ for every hyperplane $H_{ij}$. 
 The $+1$ orientation indicates that $T_i \subset H^{+} \cup H,\, T_j \subset H^{-} \cup H$ (and vice-versa).
Second, there may exist some points which lie on the hyperplanes, and hence, their assignments cannot 
be determined by their relative positions to these hyperplanes. To handle this degeneracy, 
we use the Simulation-of-Simplicity technique and assume general position.
Therefore, there are at most $p^2 \cdot k$ such ambigious points. Since this is a bounded number,
we can brute-force try all $p$ possible sets $T_1,\dots,T_p$ for such points.
This leads to a branching factor of $p^{p^2k}$. The overall branching factor is still $n^{\mathcal{O}(kp^2)}$. We now proceed to give 
a formal description as Algorithm \ref{alg:main}.

\begin{figure*}[t]
\begin{framed}\label{alg:main}
Algorithm \ref{alg:main}

\bigskip
\hrule
\bigskip

\emph{Input:} $W = \{w_1,\dots,w_n\} \subset \reals^k$, matrices $K$, $\Lambda$. \\
\emph{Output:} A partition $T_1,\dots,T_p$ of $W$ where $|T_i| = n_i$ for some fixed $n_1,\dots,n_p$. 
\begin{enumerate}
\item For every choice of ${p \choose 2}$ hyperplanes $H_{ij}$, $i\in [p],j\in[p]$ from the 
set $\mathcal{H}$ with an orientation $\sigma_{ij}\in\{+1,-1\}$,
\begin{enumerate}
\item Let $W' = \emptyset$.
\item For every $w_i\in W$ and $q \in [p]$, check if $w_i$ belongs to the convex region $R_q$
corresponding to the intersection of open halfspaces 
\[
R_q = \bigcap_{i=1, i\neq q}^p H_{qi}^{\sigma_{qi}}
\]
We use the Simulation-of-Simplicity subroutine of Section \ref{subsec:sos} to check 
the relative position of $w_i$ with respect to the hyperplanes. 
\item If $w_i$ belongs to some region $R_q$, we assign $w_i$ to the set $T_q$.
Otherwise, we add $w_i$ to the set $W'$. 
\item For every point $w_i \in W'$, try each of the $p$ assignments to $T_1,\dots,T_p$.
\item Check if the constraints $|T_i| = n_i$ are satisfied, otherwise reject this branch of computation. 
\end{enumerate}
\item For every partition $(T_1,\dots,T_p)$ computed above, evaluate the \qvp objective 
function and output an optimal solution. 
\end{enumerate}
\end{framed}
\end{figure*}

\begin{claim}\label{clm:rtime}
Given a \qvp instance, Algorithm \ref{alg:main} correctly computes an optimal solution in \rtime time. 
\end{claim}
\begin{proof}
We first show the correctness. By Lemma \ref{lem:lin-sep}, it suffices 
to show that Algorithm \ref{alg:main} computes all partitions $(T_1,\dots,T_p)$ of $[n]$
such that the family of sets $W[T_1],\dots,W[T_p]$ is weakly linearly separated. 
We claim that Algorithm \ref{alg:main} discovers every such family of sets in Step 1.
Indeed, for such a family $W[T_1],\dots,W[T_p]$, there exist $p^2$ hyperplanes $H_{ij}$ which weakly linearly separate
the sets $W[T_i],W[T_j]$, for $i,j \in [p]$. By S-o-S technique of Section \ref{subsec:sos}, we can assume general position
for the input set $W$. It can be shown that for every hyperplane $H_{ij}$,
we can equivalently find another hyperplane $\tilde{H_{ij}}$ 
in $\mathcal{H}$ with the following property. If $(A,B)$ is a partition of $[n]$ such that 
$H_{ij}$ weakly linearly separates $W[A],W[B]$, then $\tilde{H_{ij}} \in \mathcal{H}$ also weakly linearly separates $W[A],W[B]$. 
(Refer to Claim \ref{clm:hyperplanes} in the appendix). 
Therefore, there exists a branch of the algorithm in Step 1 such that we discover the hyperplanes $\tilde{H_{ij}}$.
Steps 1 (b)-(d) ensure that we recover the partition $W[T_1],\dots,W[T_p]$. 
%

The running time can be bounded as follows. The branching in Step 1 is bounded by ${n \choose k}^{p^2} \cdot 2^{p^2}$. 
In Step 1 (b), the number of calls to the Simulation of Simplicity subroutine is bounded by $n \cdot p \cdot p$, 
since we have $n$ points, $p$ regions and $p$ queries $(H_{q1},w_p),\dots,(H_{qp},w_p)$. 
By Lemma \ref{lem:sos}, every call to this subroutine has a cost $\bigOstar(k^{\mathcal{O}(k)})$ 
In Step 1 (c), there is an additional branching factor of $p^{|W'|}$ for brute-force assignment of points in $W'$. 
These are precisely the points which lie on some hyperplane $H_{ij}$, and hence $|W'| \leq p^2 \cdot k$. 
This incurs an additonal $p^{p^2k}$ branching. The remaining steps are usual polynomial time computations.
The overall running time is thus bounded by ${n \choose k}^{p^2} \cdot 2^{p^2} \cdot p^{p^2k} \cdot \bigOstar(k^{\mathcal{O}(k)})$ $\leq$ \rtime.

\end{proof}

Finally, we summarize this section with the proof of our main theorem. 

\begin{proofof}{Theorem \ref{thm:main}}
Lemma \ref{lem:reduce} and Claim \ref{clm:rtime} together imply the proof.
\end{proofof}

\section{Conclusion} 

Through our results, we were able to gain insight into the tractibility of the problems $\SIM$ and $\MSIM$. 
However, there are a few open threads which remain elusive. 
The regime of bounded rank $k$ and unbounded clustering number $p$ is 
still not fully understood for $\MSIM$, in the case of positive semi-definite matrices.
It is not clear whether the problem is $\textsf{P}$-time or \NP-hard in this case. 
Indeed, an $n^{O(k)}$ algorithm for $\MSIM$, in the case of positive semi-definite matrices, remains a possibility.
From the perspective of parameterized complexity, we can ask 
if $\MSIM$ is \textsf{W[1]}-hard, where the parameter of interest is 
the rank $k$. Finally, the approximability for the problems $\MSIM$
deserves further examination, especially for the case of bounded rank.

\appendix
\section{Appendix}

\section*{Proofs in Section \ref{sec:hard}}

\begin{proofof}{Theorem \ref{thm:hard:graphs}}
The proof is done by reduction from the \NP-hard Hamiltonian
Cycle problem in $3$-regular graphs (\textsc{Ham-Cycle});
see \cite{GaJo1979}.
Given a $3$-regular graph $G$ on $n$ vertices as an instance
of \textsc{Ham-cycle}, the reduction computes the $n$-vertex
cycle $C_n$ and graph $G$ as inputs for $\SIM$.
We recall from Section \ref{sec:prelims} that the squared
Frobenius distance $\|A_{C_n}^{\pi}-A_G\|_F^2$ between these
two graphs equals twice the number of $\pi$-mismatches.
Since $C_n$ and $G$ have $n$ and $\frac{3n}{2}$ edges,
respectively, there are at least $\frac{3n}{2}-n=\frac{n}{2}$
mismatches for any $\pi\in S_n$.
We claim that $G$ has a Hamiltonian cycle if and only if there
exists a $\pi$ for which the number of $\pi$-mismatches is
exactly $\frac{n}{2}$.
Indeed, if $G$ has a Hamiltonian cycle, the natural bijection
$\pi:V(C_n) \rightarrow V(G)$ will cause exactly $\frac{n}{2}$
mismatches.
Conversely, if there exists a $\pi$ for which the number of
mismatches is exactly $\frac{3n}{2}-n$, it must map every edge
of $C$ onto an edge of $G$.
Hence, $G$ has a Hamiltonian cycle.
\end{proofof}

\begin{proofof}{Theorem \ref{thm:hard:trees}}
The proof is by a reduction from the following \NP-hard variant
of the \textsc{Three-Partition} problem \cite{GaJo1979}, which is defined
as follows.
The input consists of integers $A$ and $a_1,\dots,a_{3m}$ in
unary representation, with $\sum_{i=1}^{3m}a_i=mA$ and with
$A/4<a_i<A/2$ for $1\le i\le3m$.
The question is to decide whether $a_1,\dots,a_{3m}$ can
be partitioned into $m$ triples so that the elements in
each triple sum up to precisely $A$.

We first show that the restriction of $\SIM$ to forests is \NP-hard.
Given an instance of \textsc{Three-Partition}, we compute an instance
of $\SIM$ on the following two forests $F_1$ and $F_2$.  
Forest $F_1$ is the disjoint union of $3m$ paths with
$a_1,\dots,a_{3m}$ vertices, respectively.
Forest $F_2$ is the disjoint union of $m$ paths that each
consists of $A$ vertices.
We claim that the \textsc{Three-Partition} instance has answer YES,
if and only if there exists a permutation $\pi$ such that 
there are at most $2m$ mismatches. 
If the desired partition exists, then for each triple we
we can pack the three corresponding paths in $F_1$ into 
one of the paths in $F_2$ with two mismatches per triple.
Conversely, if there exists a permutation $\pi$ with at 
most $2m$ mismatches, then these $2m$ mismatches cut the 
paths in $F_2$ into $3m$ subpaths (we consider 
isolated vertices as paths of length $0$).
As each of these $3m$ subpaths must be matched with a
path in $F_1$, we easily deduce from this a solution for 
the \textsc{Three-Partition} instance.

To show that $\SIM$ is \NP-hard for the class of trees, 
we modify the above forests $F_1$ and $F_2$ into trees 
$T_1$ and $T_2$. 
Formally, we add a new vertex $v_1$ to $V(F_1)$ and
then connect one end-point of every path in $F_1$ to 
$v_1$ by an edge; note that the degree of vertex $v_1$ 
in the resulting tree is $3m$.
Analogously, we add a new vertex $v_2$ to $V(F_2)$,
connect it to all paths, and thus produce a tree in
which vertex $v_2$ has degree $m$.
For technical reasons, we furthermore attach $8m$    
newly created leaves to every single vertex in 
$V(F_1)$ and $V(F_2)$. k
The resulting trees are denoted $T_1$ and $T_2$, respectively.

We claim that the considered \textsc{Three-Partition} instance 
has answer YES, if and only if there exists 
$\pi:V(T_1)\rightarrow V(T_2)$ with at most $4m$ mismatches.
If the desired partition exists, the natural bijection maps 
every original forest edge in $T_1$ to an original forest 
edge in $T_2$, except for some $2m$ out of the $3m$ edges 
that are incident to $v_1$ in $T_1$; this yields a total
number of $2m+2m=4m$ mismatches. 
Conversely, suppose that there exists a permutation $\pi$ 
with at most $4m$ mismatches.
Then $\pi$ must map $v_1$ in $T_1$ to $v_2$ in $T_2$, since 
otherwise we pay a penalty of more than $4m$ mismatches
alone for the edges incident to the vertex mapped into $v_2$.
As the number of mismatches for edges incident to $v_1$ and
$v_2$ amounts to $2m$, there remain at most $2m$ further
mismatches for the remaining edges.
Similarly as in our above argument for the forests, these
at most $2m$ mismatches yield a solution for the 
\textsc{Three-Partition} instance.
\end{proofof}

\begin{proofof}{Theorem \ref{thm:path-tree}}
If $G$ is a path and $H$ is a tree, $\SIM$ boils down to
the problem of finding a system of disjoint paths in the
tree $H$ that contains the maximal number of edges.
We root the tree $H=(V,E)$ at an arbitrary vertex, and for
every $v\in V$ we let $H(v)$ denote the induced maximal
sub-tree of $H$ that is rooted at $v$.
For $v\in V$, we let $A(v)$ denote the maximal number of
edges that can be covered by a system of disjoint paths
in tree $H(v)$.
Furthermore, we let $B(v)$ denote the maximal number of
edges that can be covered by a system of disjoint paths
in tree $H(v)$ subject to the condition that one of
these paths starts in vertex $v$.
For a leaf $v$ in $H$, we have $A(v)=B(v)=0$.
For non-leaves $v$ in $H$, a straightforward dynamic
programming approach computes $A(v)$ and $B(v)$ in linear
time from the corresponding $A$-values and $B$-values
for the children of $v$.
All in all, this yields a polynomial time algorithm.
\end{proofof}

\begin{proofof}{Theorem \ref{thm:hard:matrices}}
The proof is by a reduction from the \NP-hard \textsc{Partition} 
problem \cite{GaJo1979}, defined as follows. 
Given a set $S$ of $2n$ positive integers $\{a_1,\dots,a_{2n}\}$, 
where $a_1+ \dots+a_{2n} = 2A$, decide whether there exists a 
subset $I\subseteq\{1,\ldots,2n\}$ with $|I|=n$ such that 
$\sum_{i \in I } a_i = A$.   
We construct the following $2n\times 2n$ real symmetric 
matrices $C$ and $B$ as our $\MSIM$ instance.
The matrix $C$ is defined as $C_{ij} := a_i \cdot a_j$. 
The matrix $B$ is defined as 
\[
  B_{ij}  := \begin{cases}
              -1 &  i \in [1,n],~ j \in [1,n] \\
              -1 &  i \in [n+1,2n],~ j \in [n+1,2n] \\
             ~~0 & \mbox{otherwise} 
             \end{cases}
\]
Indeed, $\|C^{\pi} - B \|_F^2  = \|C^{\pi}\|_F^2 + \|B\|_F^2 - 2 \tracepd{C^{\pi}}{B}$.
Since $\|C^{\pi}\|_F^2 = \|C\|_F^2$ does not depend on $\pi$, it suffices to 
minimize the term $(-1)\tracepd{C^{\pi}}{B}$. The term
\begin{align*}
\tracepd{C^{\pi}}{B} 
  & = \displaystyle\sum_{i,j \in [2n]} c_{i^{\pi} j^{\pi}} b_{i j} \\ 
  & =  \displaystyle\sum_{i,j \in [1,n]} c_{i^{\pi} j^{\pi}} (-1) +  \displaystyle\sum_{i,j \in [n+1,2n]} c_{i^{\pi} j^{\pi}} (-1) \\
  & =  (-1) \left(\displaystyle\sum_{i,j \in [1,n]} c_{i^{\pi} j^{\pi}}  +  \displaystyle\sum_{i,j \in [n+1,2n]} c_{i^{\pi} j^{\pi}} \right) \\
  & = (-1) \left(  \left(\displaystyle\sum_{i \in [1,n]} a_{i^{\pi}} \right)^2  +  \left(\displaystyle\sum_{i \in [n+1,2n]} a_{i^{\pi}}\right)^2 \right).
 \end{align*}
Let $S_1 = \{ i^{\pi} \,\vert\, i \in [1,n] \}$. 
Let $S_2 = \{ i^{\pi} \,\vert\, i \in [n+1, 2n] \}$.  
Let $X_1,X_2$ be the sum of elements corresponding to the sets $S_1,S_2$ respectively. 
Clearly, $X_2 = 2A -X_1$. Then,
\begin{align*}
(-1) \tracepd{C^{\pi}}{B} &=  \left(\displaystyle\sum_{i \in S_1 } a_i \right)^2 + \left(\displaystyle\sum_{i \in S_2 } a_i\right)^2 \\
& = X_1^2 +  (2A - X_1)^2 \\ 
& \geq 2A^2
\end{align*}
using the inequality $\frac{x_1^2 + x_2^2}{2} \geq (\frac{x_1 + x_2}{2})^2$ for $x_1,x_2 \geq 0$.  
Moreover, equality is attained only for $x_1=x_2$, which implies $X_1 = 2A - X_1$, and hence $X_1=X_2=A$. 
Therefore, the given \textsc{Partition} instance has a partition of the desired kind if and only if there 
exists a $\pi$ such that $(-1)\tracepd{C^{\pi}}{B}$ attains the minimum value $2A^2$. 
Hence, the problem of minimizing $(-1)\tracepd{C^{\pi}}{B}$, and consequently 
$\|C^{\pi} - B \|_F^2$, over $\pi \in S_n$ must be \NP-hard. 
 
Finally, we show that $C$ and $B$ are matrices of rank $1$ and rank $2$ respectively.
The matrix $C$ can be expressed as a rank $1$ matrix $uu^T$, where $u = [ a_1 \dots a_{2n}]$ is a column vector of length $2n$.  
The corresponding Eigenvalue can be checked to be $\|u\|^2 = (a_1^2 + \dots + a_{2n}^2)$. In particular, $A$ is positive semi-definite.
The matrix $B$ can be expressed as the sum of two rank-$1$ matrices $B_1+B_2$ where (a) $B_1 = (-1) v_1v_2^T$, $B_2 = (-1) v_2v_2^T$  
and (b) $v_1$ is a $0$-$1$ column vector of length $2n$ such that the $i^{th}$ coordinate of $v_1$ is $1$ iff $1\leq i\leq n$. 
Similarily, $v_2$ is a $0$-$1$ column vector of length $2n$ such that the $i^{th}$ coordinate of $v_2$ is $1$ iff $n+1\leq i\leq 2n$. 
The corresponding Eigenvalues can be checked to be $-\|v_1\|^2,-\|v_2\|^2$ which is the multiset $\{-n,-n \}$.
\end{proofof}

\begin{proofof}{Theorem \ref{thm:hard:psd}}
 In the proof of Theorem \ref{thm:hard:graphs}, instead of considering the adjacency matrices 
 of $C_n$ and $G$, we consider their Laplacian matrices $L_{C_n}$ and $L_G$. 
 Since $L_{C_n} = D_{C_n} - A_{C_n}$ and $C_n$ is $2$-regular, $L_{C_n} = 2 I_n - A_{C_n}$. 
 Since $L_{G} = D_{G} - A_{G}$ and $G$ is $3$-regular, $L_{G} = 3 I_n - A_{G}$. 
 Therefore, the quantity 
 \begin{align*}
 \|L_{C_n}^{\pi} - L_G\|_F^2  & =  \|(2 I_n - A_{C_n})^{\pi} - (3 I_n - A_{G}) \|_F^2 \\
 & =  \|2 I_n - A_{C_n}^{\pi} - (3 I_n - A_{G}) \|_F^2 \\
 & =  \|- I_n - (A_{C_n}^{\pi} - A_{G}) \|_F^2 \\
 & =  \| (A_{C_n}^{\pi} - A_{G}) + I_n\|_F^2 \\
 & =  \| (A_{C_n}^{\pi} - A_{G})\|_F^2 + \| I_n\|_F^2 \\
 & =  \| (A_{C_n}^{\pi} - A_{G})\|_F^2 + n.
 \end{align*}
The second last equality follows because $I_n$ has only diagonal entries whereas every diagonal entry of $(A_{C_n}^{\pi} - A_{G})$ is zero. 
The above calculation shows that these two quantites differ by $n$ (which is independent of $\pi$). Therefore, computing the 
Frobenius distance between the two Laplacian matrices $L_{C_n}$ and $L_{G}$ is \NP-hard as well. 
\end{proofof}

\begin{proofof}{Theorem \ref{thm:hard:psd:bnd}}
 We modify the hardness proof for Theorem \ref{thm:hard:matrices}. We define the matrix $C$ to be 
 the same as in the proof of Theorem \ref{thm:hard:matrices}. We define the matrix $B':=B+nI_n$
 where $B$ is the matrix from the proof of Theorem \ref{thm:hard:matrices}. 
 Since the Eigenvalues were shown to be are $-n$ with multiplicity $2$ and $0$ with multiplicity $n-2$,
 adding the matrix $nI_n$ to $B$ shifts the Eigenvalues by $+n$, and hence $B'$ is p.s.d. It remains to observe that 
 the significant quantity $\tracepd{C^{\pi}}{B'}$ differs from the corresponding $\tracepd{C^{\pi}}{B'}$ by a constant independent of $\pi$.
 Indeed,
 \begin{align*}
 \tracepd{C^{\pi}}{B'} &= \displaystyle\sum_{i,j \in [2n]} c_{i^{\pi} j^{\pi}} b'_{i j} \\ 
		       &= \displaystyle\sum_{i,j \in [2n]} c_{i^{\pi} j^{\pi}} b_{i j} + \displaystyle\sum_{i \in [2n]} c_{i^{\pi} i^{\pi}} n \\ 
 &= \displaystyle\sum_{i,j \in [2n]} c_{i^{\pi} j^{\pi}} b_{i j} + \displaystyle\sum_{i \in [2n]} a_i^2 n \\ 
&= \tracepd{C^{\pi}}{B} + n \left(\displaystyle\sum_{i \in [2n]} a_i^2\right). \\ 
 \end{align*}
 Hence, the problem of minimizing $\|C^{\pi} - B' \|_F^2$ over $\pi \in S_n$ must be \NP-hard. 
 Recall that the matrix $C$ was shown to be positive semi-definite in the proof of Theorem \ref{thm:hard:matrices}.
 This finishes the proof of our theorem.
\end{proofof}

\section*{Proofs in Section \ref{sec:algo}}

\begin{proofof}{Lemma \ref{lem:reduce}}
The spectral 
decompositions of $A$ and $B$ are represented by $\specA$ and $\specB$. 
Since the cluster-number of $B$ is $p$, let $\tilde{V} = \{\tilde{V}^1,\dots,\tilde{V}^p\}$
be the set of distinct vectors in the multiset $\rows(V)$.
Let $n_1,\dots,n_p$ be the multiplicity of the  
elements $\tilde{V}^1,\dots,\tilde{V}^p$ respectively. Clearly, $n_1+ \dots + n_p = n$.
Let $\tilde{P}$ be the natural partition arising from this clustering. In other words, 
$\tilde{P} = S_1 \cup \dots \cup S_p$ be a partition of $[n]$ where $S_l = \{ i \,\vert\, V^{i} = \tilde{V}^l\}, l \in [p]$.

 Let $\Pi^*$ denote the set $\argmin_{\pi}  \frob{A^{\pi} - B}$. We first restate $\Pi^*$ as follows. 
Observe that $\frob{A^{\pi} - B }^2 = \tracepd{A^{\pi}-B}{A^{\pi}-B} = \tracepd{A^{\pi}}{A^{\pi}} + \tracepd{B}{B} - 2 \tracepd{A^{\pi}}{B}$. 
Since $\tracepd{A^{\pi}}{A^{\pi}} = \frob{A^{\pi}}^2 = \frob{A}^2 = \tracepd{A}{A}$, we have $\frob{A^{\pi}-B}^2 = \tracepd{A}{A} + \tracepd{B}{B} - 2 \tracepd{A^{\pi}}{B}$.
Therefore, we can equivalently maximize $\tracepd{A^{\pi}}{B}$ over $\pi \in \Sn$. We have 
\begin{align*}
 \Pi^* & = \argmin_{\pi}  \frob{A^{\pi} - B} \\ 
       & = \argmax_{\pi}  \tracepd{A^{\pi}}{B} \\ 
       & = \argmax_{\pi}  \displaystyle\sum_{i,j \in [n]} a_{i^{\pi}j^{\pi}} b_{ij} \\ 
       & =  \argmax_{\pi}  \displaystyle\sum_{i,j \in [n]} \innerpd{U^{i^{\pi}}}{\Lambda}{U^{j^{\pi}}} \,\cdot\,  \innerpd{V^{i}}{\Gamma}{V^{j}}. 
\end{align*}
 
Restating $\Pi^*$ further, we get 
\begin{align*}
\Pi^* &= \argmax_{\pi} \displaystyle\sum_{i,j \in [n]} \innerpd{U^{i^{\pi}}}{\Lambda}{U^{j^{\pi}}} \,\cdot\,  \innerpd{V^{i}}{\Gamma}{V^{j}} \\
      &= \argmax_{\pi} \displaystyle\sum_{l,m \in [p]} \left( \displaystyle\sum_{i\in S_l,\, j \in S_m} \innerpd{U^{i^{\pi}}}{\Lambda}{U^{j^{\pi}}} \,\cdot\,  \innerpd{V^{i}}{\Gamma}{V^{j}}\right)  \\
      &= \argmax_{\pi} \displaystyle\sum_{l,m \in [p]} \left( \displaystyle\sum_{i\in S_l,\, j \in S_m} \innerpd{U^{i^{\pi}}}{\Lambda}{U^{j^{\pi}}} \right)\,\cdot\,  \innerpd{\tilde{V}^l}{\Gamma}{\tilde{V}^m} \\
      &= \argmax_{\pi} \displaystyle\sum_{l,m \in [p]} \left\langle \sum_{i\in S_l} U^{i^{\pi}}, \sum_{j\in S_m} U^{j^{\pi}}  \right\rangle_{\Lambda}  \,\cdot\,  \innerpd{\tilde{V}^l}{\Gamma}{\tilde{V}^m} \\
      &= \argmax_{\pi} \displaystyle\sum_{l,m \in [p]} \left\langle U^{S_l^{\pi}},U^{S_m^{\pi}}\right\rangle_{\Lambda} \,\cdot\,  \innerpd{\tilde{V}^l}{\Gamma}{\tilde{V}^m} 
\end{align*}
where we recall the notation $U^S = \sum_{i\in S} U^i$, $S \subseteq [n]$.  
Let $K$ be the $p\times p$ matrix defined as $K_{lm} = \innerpd{\tilde{V}^l}{\Gamma}{\tilde{V}^m}$. Clearly, $K$ is positive semi-definite.
Simplifying, we obtain 
\[
 \Pi^* = \argmax_{\pi \in \Sn} \displaystyle\sum_{l,m \in [p]} K_{lm} \left\langle  U^{S_l^{\pi}},U^{S_m^{\pi}}\right\rangle_{\Lambda}.
\]
Given a permutation $\pi \in \Sn$, we can bijectively associate a partition $P_{\pi} = (S_1^{\pi} \cup \dots \cup S_p^{\pi}) \in \mathcal{P}$. 
Recall that $(S_1,\dots,S_p)$ is the partition corresponding to the clustering of $\rows(V)$. 
Therefore, the set $\Pi^*$ is in one-to-one correspondence with the set 
\[
 \Omega^*= \argmax_{P \in \mathcal{P}} \displaystyle\sum_{l,m \in [p]} K_{lm} \left\langle U^{T_l},U^{T_m}\right\rangle_{\Lambda}.
\]
Clearly, this is an instance of \qvp with the input set $\rows(U) \subseteq \reals^k$
of size $n$, along with the corresponding matrices $K$ and $\Lambda$. This instance can be computed directly from the spectral decompositions of $A$ and $B$, 
which can be done in $\bigOstar(1)$ time. Moreover, given an optimal solution for this \qvp instance, we can uniquely recover an optimal permutation $\pi \in \Pi^*$  in $\bigO(1)$ time. Hence, proved. 

\end{proofof}

\begin{proofof}{Claim \ref{clm:lin-sep1}}
\begin{align*}
\Omega^* &= \argmax_{P \in \mathcal{P}} \displaystyle\sum_{l,m \in [p]} K_{lm} \left\langle W^{T_l},W^{T_m}\right\rangle_{\Lambda} \\ 
& =  \argmax_{P \in \mathcal{P}} \displaystyle\sum_{l,m \in [p]} K_{lm}  \left( \displaystyle\sum_{q=1}^{k}  \lambda_q (W^{T_l})_q (W^{T_m})_q \right) \\
				  &= \argmax_{P \in \mathcal{P}}  \displaystyle\sum_{q=1}^{k} \lambda_q \left(  \displaystyle\sum_{l,m \in [p]}  (W^{T_l})_q \cdot K_{lm} \cdot (W^{T_m})_q \right) \\
				  & = \argmax_{P \in \mathcal{P}}  \displaystyle\sum_{q=1}^{k} \lambda_q \innerpd{X_q}{K}{X_q}.  				 
\end{align*}
\end{proofof}

\begin{proofof}{Claim \ref{clm:lin-sep2}}
 Let $F = \displaystyle\sum_{q=1}^{k} \lambda_q \innerpd{X_q}{K}{X_q}$ denote the objective function of Claim \ref{clm:lin-sep1}. Let $G_q(X)$ denote the function 
 $\lambda_q \innerpd{X}{K}{X}$. Since $\lambda_i>0$,
 Lemma \ref{lem:psdconv} implies that $G_q$ is a convex function for $q\in [p]$. Applying Corollary \ref{coro:conv-lin} for $G_1,\dots,G_k$ finishes the proof.
\end{proofof}

\begin{proofof}{Claim \ref{clm:lin-sep3}}
  \begin{align*}
\displaystyle\sum_{q=1}^{k} \lambda_q \innerpd{X^*_q}{K}{X_q} &= \displaystyle\sum_{q=1}^{k} \lambda_q  \left( \displaystyle\sum_{l,m=1}^{p} (X^*_q)_l \,K_{lm}\,(X_q)_m \right) \\
                                                     &= \displaystyle\sum_{l,m=1}^{p}  K_{lm} \left(\displaystyle\sum_{q=1}^{k} \lambda_q (X^*_q)_l (X_q)_m \right)\\
						     &= \displaystyle\sum_{l,m=1}^{p} K_{lm} \innerpd{W^{T_l^*}}{\Lambda}{W^{T_m}} \\
						     & = \displaystyle\sum_{l=1}^{p} \innerpd{\mu_l}{}{W^{T_m}}.
\end{align*}
for some vectors $\mu_1,\dots,\mu_p \in \reals^k$.
\end{proofof}

\begin{claim}\label{clm:hyperplanes}
Let $W$ be a set of $n$ points $\{w_1,\dots,w_n\} \subset \reals^k$ in general position,
 where $n>k$.
Suppose $W_1,W_2$ is a weakly linear separation of $W$ by a hyperplane $H$. 
Then, there exists another hyperplane $\tilde{H}$ with the following properties:
(a) $\tilde{H}$ passes through exactly $k$ points of $W$, and (b) $\tilde{H}$ also weakly linearly separates $W_1,W_2$. 
\end{claim}
\begin{proof}
 Let $S\subseteq W$ be the set of points in $W$ which already lie on $H$. Since $W$ is in general position, 
 $|S|\leq k$. If $|S|=k$ already, we are done. Otherwise, $|S|=l<k$. 
 Let $S=\{w_{i_1},\dots,w_{i_l}\}$. Moreover, let $H$ be represented by the linear equation $c^Tx = \alpha$, 
 where $c \in \reals^k$ and $\alpha\in\reals$. Since $l<k$, there exists a vector $\delta \in \reals^k$ satisfying the system of $l$ linear equations $\delta^T w_{i_1} = \dots = \delta^T w_{i_l} = 0$.
 Let $\gamma \in \reals$. Consider the hyperplane $H_{\gamma} := (c + \gamma \delta)^T x =\alpha$. Clearly, $S$ lies on $H_\gamma$. 
 We select $\gamma$ suitably as follows. We slowly increase (or decrease) the value of $\gamma$ from zero such that the hyperplane $H_{\gamma}$ hits a point $w \in W \setdiff S$ for the first time. 
 Therefore, we obtain a new hyperplane $H_{\gamma}$ such that the set $S \cup \{w\}$ lies on $H'$. Moreover, it is easy to check that (a) for every point $w \notin S \cup \{w\}$, 
 the relative position w.r.t $H_{\gamma}$ is same as the relative position w.r.t $H$ (in terms of the halfspaces $H^+$, $H^-$) and hence, (b)
 if $W_1,W_2$ is a weak linear separation of $W$ by $H$, it remains a weak linear separation of $H'$.  
 Repeating this argument, we ultimately obtain a hyperplane $\tilde{H}$ which passes through a set $W'$ of size $k$, satisfying $W \supset W'\supset S$.
 Hence, proved.

\end{proof}

\bibliographystyle{plain}
\bibliography{rgi}

\end{document}